\documentclass[conference]{IEEEtran}
\usepackage[margin=0.9in,top=0.7in]{geometry}

\usepackage[utf8]{inputenc}

\usepackage{cite}
\usepackage[pdftex]{graphicx}
\graphicspath{{Figure/}}
\usepackage[cmex10]{amsmath}
\usepackage{amsthm}
\usepackage{amssymb}
\usepackage{algorithmic}
\usepackage{array}
\usepackage{color}
\usepackage{epstopdf}
\usepackage{amsfonts}

\usepackage{pgfplots}
\pgfplotsset{compat=1.3}
\usepackage{tikz}							
\usetikzlibrary{shapes}
\usetikzlibrary{spy}
\usetikzlibrary{circuits}
\usetikzlibrary{arrows}
\usetikzlibrary{spy}
\usepackage{subfig}

\usepackage[ruled,vlined]{algorithm2e}

\newcommand{\vect}[1]{\boldsymbol{\mathrm{#1}}}
\newcommand{\mat}[1]{\boldsymbol{\mathrm{#1}}}

\newcommand{\tr}{\mathrm{tr}}

\newcommand{\diag}{\text{diag}}

\usepackage{mathtools}

\theoremstyle{definition}
\newtheorem{theorem}{Theorem}
\newtheorem{proposition}{Proposition}

\newtheorem{corollary}{Corollary}

\begin{document}

\linespread{1}
\title{Optimal Downlink Training Sequence \\for Massive MIMO Secret-Key Generation\vspace{-1em}}

\author{\IEEEauthorblockN{Fran\c{c}ois Rottenberg\IEEEauthorrefmark{1}\IEEEauthorrefmark{2}
	}
	
\IEEEauthorblockA{\IEEEauthorrefmark{1}Universit\'{e} catholique de Louvain, Louvain-la-Neuve, Belgium,
	}

\IEEEauthorblockA{\IEEEauthorrefmark{2}Universit\'{e} libre de Bruxelles, Brussels, Belgium.}\vspace{-1.5em}
}


\maketitle

\begin{abstract}
	In this paper, the secret-key capacity is maximized by optimizing the downlink training sequence in a time division duplexing (TDD) massive multiple-input-multiple-output (MIMO) scenario. Both single-user and multiple user cases are considered. As opposed to previous works, the optimal training sequence and the related secret-key capacity is characterized in closed-form in the single-user case and the large antenna multiple-user case. Designs taking into account a constraint on the maximal number of pilots are also proposed. In the multiple-user case, both the max-min and the sum capacity criteria are considered, including potential user priorities. In the end, it is shown that massive MIMO boosts the secret-key capacity by leveraging: i) spatial dimensionality gain and ii) array gain. Moreover, in the large antenna case, the multiple-user capacity is obtained with no extra pilot overhead as compared to the single-user case.
\end{abstract}

\begin{IEEEkeywords}
Secret-key generation, Massive MIMO, training design.
\end{IEEEkeywords}

\section{Introduction}\label{section:Introduction}
\linespread{1}

Secret-key generation based on wireless channel reciprocity is an interesting alternative to cryptographic primitives as it can be efficiently implemented at the physical layer of emerging wireless communication networks, while providing information-theoretic security guarantees \cite{Csiszar1993,Maurer1993}. In particular, secret-key generation is particularly promising in massive multiple-input-multiple-output (MIMO) scenarios, which are being heavily deployed in 5G \cite{Jiao2019review,Li2019review}. Massive MIMO allows to boost the secret-key capacity, \textit{i.e.}, the maximal rate at which secret bits can be generated. This is obtained by leveraging spatial dimensionality gains together with array gains.

Different works have already considered physical key generation for massive MIMO systems. For instance, the work of \cite{Im2015} studied the impact of pilot contamination attack. The authors of \cite{Jiao2018perturbationangle} suggested to use perturbations of angle to increase the secret-key rate. In \cite{Jiao2018}, the authors looked at secret-key generation in millimeter-wave massive MIMO. Training design was also considered in \cite{Sun2020IEEEAccess,Chen2020,Li2021}. 

In contrast to previous approaches in the domain, this work provides a detailed study of the DL (downlink) training sequence design to optimize the massive MIMO secret-key capacity in time divison duplexing (TDD). In contrast to previous approaches, closed-form solutions are provided in both the single-user case and the multiple-antenna case. It is shown how the large number of antennas at the base (BS) allows to boost the capacity, while reducing the pilot overhead. Indeed, one of the key findings is that the required number of pilots to achieve the capacity does not scale with the number of users as long as the number of antennas at the BS is large enough. This implies that the pilot overhead remains limited, even in multiple-user scenarios. Designs that take into account a constraint on the maximal number of pilots to be sent are also proposed. Finally, in the multiple-case, different optimization criteria are considered such as the maximization of the sum capacity and the max-min capacity, including potential user priorities.

\textbf{Notations}: 
Vectors and matrices are denoted by bold lowercase and uppercase letters $\vect{a}$ and $\mat{A}$, respectively (resp.). Superscripts $^*$, $^T$, $^H$ and $^{\dagger}$ stand for conjugate, transpose, Hermitian transpose and Moore-Penrose pseudo-inverse. The symbols $\tr[.]$, $\mathbb{E}(.)$, $\Im(.)$ and $\Re(.)$ denote the trace, expectation, imaginary and real parts, respectively. $\jmath$ is the imaginary unit. $\|\mat{A}\|$ and $|\mat{A}|$ are the Frobenius norm and determinant respectively. $\mat{I}_N$ denotes the identity matrix of order $N$. $\mat{0}_{N\times M}$ is a zero matrix of size $N\times M$. Subscripts of matrices are dropped whenever matrix dimensions are clear from the context. 
$\diag(\vect{a})$ returns a diagonal matrix with $\vect{a}$ on its diagonal. The positive part of a real quantity is denoted by $[a]^+=\max(a,0)$. $\sigma_n(\mat{A})$ (or $\lambda_n(\mat{A})$) is the $n$-th largest singular (or eigenvalue) of $\mat{A}$. $\otimes$ stands for the Kronecker product. $\mat{A}^{1/2}$ denotes the square root of $\mat{A}$, uniquely defined for positive semidefinite matrix $\mat{A}$.

\section{System Model and Secret-Key Capacity}
\label{section:system_model}

\begin{figure}[!t]  
	\centering
	
	\resizebox{0.8\textwidth}{!}{%
		{\includegraphics[clip, trim=0cm 12.5cm 8cm 0cm, scale=1]{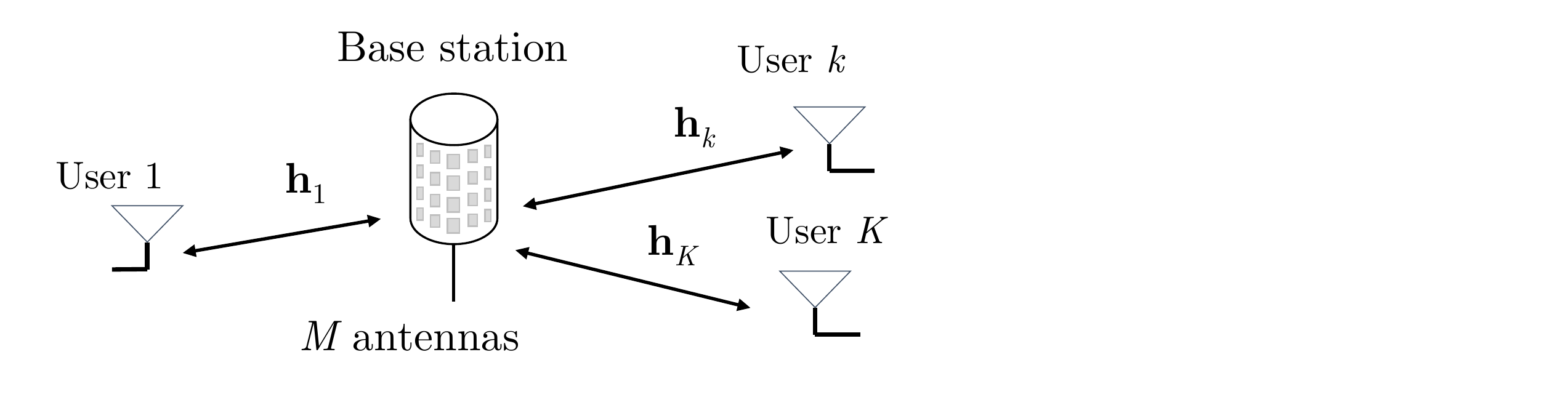}} 
	}
		\vspace{-1em}
	\caption{Massive MIMO channel model.}
	\label{fig:channel_model}
	\vspace{-1.5em}
\end{figure}

A massive MIMO scenario is considered with $K$ single-antenna users and a BS equipped with $M$ antennas, as shown in Fig.~\ref{fig:channel_model}. The channel is assumed to be frequency flat, which is typically the case for a given subchannel of an orthogonal frequency division multiplexing (OFDM) system. The channel vector from user $k$ to the BS is denoted by $\vect{h}_k\in \mathbb{C}^{M\times 1}$. The channels from each user are assumed to be independent and can be modeled as a zero mean circularly symmetric Gaussian (ZMCSG) vector with covariance matrix $\mat{R}_k \in \mathbb{C}^{M\times M}$. The BS has knowledge of the covariance matrices $\mat{R}_k$, which can be easily estimated in practice if the channel is stationary, averaging over long term statistics of $\vect{h}_k$. Hence, it is not a source of randomness used to generate a secret key. The following \textit{compact}\footnote{By compact, it is implied that only non zero eigenvalues and related eigenvectors are kept.} eigenvalue decomposition is defined $\mat{R}_k=\mat{Q}_k\mat{\Lambda}_k\mat{Q}_k^\dagger$, where $\mat{\Lambda}_k=\diag(\lambda_{1,k},...,\lambda_{S_k,k}) \in \mathbb{C}^{S_k \times S_k}$ is a diagonal matrix with strictly positive diagonal elements and $\mat{Q}_k \in \mathbb{C}^{M \times S_k}$ contains eigenvectors corresponding to the $S_k$ non zero spatial modes. Using previous notations, the channel $\vect{h}_k$ can be equivalently written as
\begin{align}
\vect{h}_k=\mat{Q}_k\mat{\Lambda}_k^{1/2}\vect{g}_k \label{eq:equivalent_h_g},
\end{align}
where $\vect{g}_k\in \mathbb{C}^{S_k\times 1}$ is a ZMCSG with zero mean and an identity covariance matrix.

A conventional TDD system is considered so that channel reciprocity holds between the uplink (UL) and DL channels. In the UL, the pilot overhead remains limited as it scales with the number of antennas at the transmit side $K$. As common in massive MIMO networks, users transmit orthogonal UL pilots\footnote{For instance, by transmitting $K$ pilots at $K$ successive OFDM symbols.}, so that the BS can get a channel estimate for each user as
\begin{align}
\vect{z}_{k,\mathrm{UL}} &= p_{k,\mathrm{UL}}^{1/2} \vect{h}_k + \vect{w}_{k,\mathrm{UL}} \label{eq:initial_MU}\\
&=p_{k,\mathrm{UL}}^{1/2} \mat{Q}_k\mat{\Lambda}_k^{1/2}\vect{g}_k + \vect{w}_{k,\mathrm{UL}},\nonumber
\end{align}
where $p_{k,\mathrm{UL}}$ is the pilot power of user $k$, $\vect{w}_{k,\mathrm{UL}}$ is ZMCSG additive noise of covariance $\sigma_{\mathrm{UL}}^2 \mat{I}_M$.

In conventional TDD massive MIMO, the BS can rely on reciprocity to acquire DL channel state information (CSI) from UL pilots to perform, \textit{e.g.}, beamforming. Hence, no additional DL pilots are necessary\footnote{In practice, only user-specific reference signals are transmitted to allow for coherent detection at the user side.}. However, for secret-key generation, the users need to be able to estimate (at least a part of) the CSI as the BS, so as to generate a common key. A massive MIMO network is typically characterized by $K<<M$, implying that the pilot overhead, which scales with the number of transmit antennas, can quickly become prohibitive in DL. In practice, channel dimensionality reduction can be used, leveraging on the spatial statistics $\mat{R}_k$, to decrease the number of required pilots $T$, while keeping a large secret-key capacity. To do this, the BS transmits a training sequence $\mat{S}_{\mathrm{DL}}\in \mathbb{C}^{M\times T}$, where $T$ denotes the number of transmitted pilots, and so that the user $k$ observes
\begin{align*}
\vect{z}_{k,\mathrm{DL}} &= \mat{S}_{\mathrm{DL}}^\dagger \vect{h}_k + \vect{w}_{k,\mathrm{DL}} \\
&= \mat{S}_{\mathrm{DL}}^\dagger \mat{Q}_k\mat{\Lambda}_k^{1/2}\vect{g}_k + \vect{w}_{k,\mathrm{DL}},
\end{align*}
where $\vect{w}_{k,\mathrm{DL}} \in \mathbb{C}^{T\times 1}$ is ZMCSG additive noise of covariance $\sigma_{\mathrm{DL}}^2 \mat{I}_T$. If the environment is rich in scattering and if the eavesdropper is located more than a wavelength away from any users or the BS \cite{Rottenberg2020pimrc}, its observations are uncorrelated. For a given training matrix $\mat{S}_{\mathrm{DL}}$, the secret-key capacity of user $k$ is given by \cite{bloch2011physical}
\begin{align*}
C_k&=I(\vect{z}_{k,\mathrm{UL}};\vect{z}_{k,\mathrm{DL}})=h({\vect{z}}_{k,\mathrm{UL}})-h({\vect{z}}_{k,\mathrm{UL}}|{\vect{z}}_{k,\mathrm{DL}}).
\end{align*}
Using \cite[Lemma 1]{wong2009secret}, $C_k$ becomes
\begin{align*}
C_k&=\log \frac{\left| \mat{C}_{\vect{z}_{k,\mathrm{UL}}}\right|}{\left|\mat{C}_{\vect{z}_{k,\mathrm{UL}}}-\mat{C}_{\vect{z}_{k,\mathrm{UL}}\vect{z}_{k,\mathrm{DL}}}\mat{C}_{\vect{z}_{k,\mathrm{DL}}}^{-1}\mat{C}_{\vect{z}_{k,\mathrm{UL}}\vect{z}_{k,\mathrm{UL}}}^\dagger \right|},
\end{align*}
where $\mat{C}_{\vect{x}}=\mathbb{E}(\vect{x}\vect{x}^\dagger)$ and $\mat{C}_{\vect{x}\vect{y}}=\mathbb{E}(\vect{x}\vect{y}^\dagger)$. After evaluation of the covariance matrices and applying Woodbury matrix inversion lemma, $C_k$ becomes
\begin{align}
C_k&=\log \frac{\left|\mat{\Lambda}_k  + \rho_{k,\mathrm{UL}}^{-1} \mat{I}\right|}{\left| \rho_{k,\mathrm{UL}}^{-1} \mat{I}+ ( \mat{\Lambda}_k^{-1} + \sigma^{-2}_{\mathrm{DL}}  \mat{Q}_k^\dagger \mat{C}_{\mathrm{DL}} \mat{Q}_k)^{-1}  \right|}, \label{eq:C_k}
\end{align}
where $\rho_{k,\mathrm{UL}}=\frac{p_{k,\mathrm{UL}}}{\sigma_{\mathrm{UL}}^2}$ and $\mat{C}_{\mathrm{DL}}=\mat{S}_{\mathrm{DL}}\mat{S}_{\mathrm{DL}}^\dagger$ is a positive semidefinite matrix, which can be seen as the training covariance matrix. In the following sections, the optimization of the training sequence $\mat{S}_{\mathrm{DL}}$, or equivalently $\mat{C}_{\mathrm{DL}}\succcurlyeq \mat{0}$, is considered so as to maximize the secret-key capacity of each user under a training power constraint $\tr[\mat{C}_{\mathrm{DL}}]\leq p_{\mathrm{DL}}$.

Note that the BS and the users sample complex observations of the channel. At the price of a capacity reduction, robustness to potential phase offsets can be gained by sampling instead only the envelope or modulus of the channel observations \cite{Rottenberg2020CSIvsRSS}.

\section{Single-User Case}

In the case of a single-user, the index $k$ is dropped for clarity. The secret-key capacity is given by
\begin{align}
C&=\log \frac{\left|\mat{\Lambda}  + \rho_{\mathrm{UL}}^{-1} \mat{I}\right|}{\left| \rho_{\mathrm{UL}}^{-1} \mat{I}+ ( \mat{\Lambda}^{-1} + \sigma^{-2}_{\mathrm{DL}}  \mat{Q}^\dagger \mat{C}_{\mathrm{DL}} \mat{Q})^{-1}  \right|}, \label{eq:C}
\end{align}
and the solution of the optimization problem
\begin{align}
\max_{\mat{C}_{\mathrm{DL}}\succcurlyeq \mat{0}} \ C \quad \text{s.t. }\tr[\mat{C}_{\mathrm{DL}}]\leq p_{\mathrm{DL}}, \label{eq:max_capa_SU}
\end{align}
is given in next theorem.

\begin{theorem} \label{theorem:proof_theorem_SU}
	The following training sequence is a solution of the single-user problem (\ref{eq:max_capa_SU})
	\begin{align}
		\mat{S}_{\mathrm{DL}}^\dagger &= \mat{P}_{\mathrm{DL}}^{1/2} \mat{Q}^\dagger,\ \mat{P}_{\mathrm{DL}}^{1/2}=\diag\left(p_{1,\mathrm{DL}}^{1/2},...,p_{S,\mathrm{DL}}^{1/2}\right). \label{eq:SU_opt_training}
	\end{align}
	The power associated to each spatial mode $p_{s,\mathrm{DL}}$ is given by the water filling solution $p_{s,\mathrm{DL}}=\left(\mu -\sigma_{\mathrm{DL}}^2/{\lambda_s}\right)^+$, where $\mu$ is a positive constant that ensures that the power constraint is satisfied, $\tr[\mat{C}_{\mathrm{DL}}]= \tr[\mat{P}_{\mathrm{DL}}]= p_{\mathrm{DL}}$. The resulting secret-key capacity is given by
	\begin{align}
	C
	&= \sum_{s=1}^S \log \left(1+ \frac{\lambda_s^2\rho_{\mathrm{UL}}\rho_{s,\mathrm{DL}}}{\lambda_s(\rho_{\mathrm{UL}}+\rho_{s,\mathrm{DL}})+1 }\right)
	, \label{eq:SU_secret_key_capacity}
	\end{align}
	where $\rho_{s,\mathrm{DL}}=\frac{p_{s,\mathrm{DL}}}{\sigma_{\mathrm{DL}}^2}$.
\end{theorem}
\begin{proof}
	Due to space constraints, see extended version to appear soon.
\end{proof}
The optimal secret-key capacity in (\ref{eq:SU_secret_key_capacity}) shows the two main advantages of massive MIMO for secret-key generation: i) array gain $\lambda_s$ boosting the signal-to-noise ratio (SNR) and ii) spatial dimensionality gain $S$ through the use of multiple parallel spatial modes.

The theorem provides intuitive and practical insights for implementation. Consider that the BS uses the training matrix $\mat{S}_{\mathrm{DL}}^\dagger = \mat{P}_{\mathrm{DL}}^{1/2} \mat{Q}^\dagger$. Reciprocally, in UL, the BS applies the decoding matrix $\mat{Q}^\dagger$. This implies that ${\vect{z}}_{\mathrm{UL}}$ and ${\vect{z}}_{\mathrm{DL}}$ become
\begin{align}
\tilde{\vect{z}}_{\mathrm{UL}} &= p_{\mathrm{UL}}^{1/2}\mat{\Lambda}^{1/2}\vect{g}  + \tilde{\vect{w}}_{\mathrm{UL}}\label{eq:second_SU}\\
\tilde{\vect{z}}_{\mathrm{DL}} &=  \mat{P}_{\mathrm{DL}}^{1/2}\mat{\Lambda}^{1/2}\vect{g}+  \vect{w}_{\mathrm{DL}},\nonumber
\end{align}
where $\tilde{\vect{w}}_{\mathrm{UL}}=\mat{Q}^\dagger{\vect{w}}_{\mathrm{UL}}$ is ZMCSG additive noise of covariance $\sigma_{\mathrm{UL}}^2 \mat{I}_S$. Inspection of (\ref{eq:second_SU}) shows that the precoding/decoding operation at the BS has converted the initial problem into a set of $S$ parallel independent channels for secret-key generation. Vector equations in (\ref{eq:second_SU}) can be rewritten scalar-wise for $s=1,...,S$ as
\begin{align*}
\tilde{{z}}_{s,\mathrm{UL}} &= p_{\mathrm{UL}}^{1/2}{\lambda}^{1/2}_s{g}_s  + \tilde{{w}}_{s,\mathrm{UL}}\\
\tilde{{z}}_{s,\mathrm{DL}} &=  p_{s,\mathrm{DL}}^{1/2}{\lambda}^{1/2}_s{g}_s+  {w}_{s,\mathrm{DL}}.\nonumber
\end{align*}
The secret-key capacity is then given by the sum of the secret-key capacities of each independent channel. Optimizing over the power per spatial mode leads to the water filling solution of Theorem~\ref{theorem:proof_theorem_SU}. Thus, more power is allocated to stronger spatial modes while weaker modes might not be allocated any power leading to $T<S$ pilots\footnote{If spatial modes are inactive, the corresponding rows of $\mat{P}_{\mathrm{DL}}^{1/2}$ might be removed without impacting the capacity.}. As the DL SNR $p_{\mathrm{DL}}/\sigma_{\mathrm{DL}}^2$ grows large, the power becomes uniformly allocated to all spatial modes so that $S=T$. Somewhat surprisingly, the optimal power allocation does not depend on the UL SNR $\rho_{\mathrm{UL}}=p_{\mathrm{UL}}/\sigma_{\mathrm{UL}}^2$.

In practice, sending a total of $T=S$ pilots might be prohibitive, especially in rich scattering environments where $S$ might be close to $M$. In these situations, looking for the optimal training sequence with a reduced number of pilots might be attractive, even though penalizing the secret-key rate. 
\begin{corollary} \label{corollary:SU_pilot_constraint}
	If a constraint is added on the number of pilots $T\leq T_{\mathrm{max}}$ to the problem in (\ref{eq:max_capa_SU}), the optimal training design is again given by (\ref{eq:SU_opt_training}), except that only the $T_{\mathrm{max}}$ strongest spatial modes will be active, again according to a water filling allocation among active nodes.
\end{corollary}
\begin{proof}
	Due to space constraints, see extended version to appear soon.
\end{proof}
In the extreme case of a single pilot $T=1$, all the training power $p_{\mathrm{DL}}$ is allocated to the strongest spatial mode $\lambda_s$. In that case, the only gain of the massive MIMO array versus a single-antenna base station comes from the array gain.

As a final remark, one should note the convenience of the proposed training designs. Only the BS needs the knowledge of the channel statistics $\mat{R}$ (and thus $\mat{\Lambda}$ and $\mat{Q}$) to precode/decode pilot signals. To perform quantization, the user only needs to know the values of the products $\rho_{s,\mathrm{DL}}{\lambda}_s$ for each active spatial mode, which can be communicated using a public authenticated channel.

\section{Multiple-User Case}

The training design is in general more intricate in the multiple-user case. Indeed, in the DL, the training sequence $\mat{S}_{\mathrm{DL}}$ is received by all users and not only one. This implies that the training sequence must be optimized jointly considering the secret-key capacities of all users. Hence, the optimization criterion should be cleverly chosen. In general, no closed-form solution exists for the optimal training design. In the following sections, the general form of the solution is derived, it is shown how to find the optimal solution numerically and how a closed-form solution can be obtained in the large antenna case.

\subsection{General Form}
Let us define the matrix $\tilde{\mat{Q}}=\begin{pmatrix}
\mat{Q}_1,\hdots,\mat{Q}_K
\end{pmatrix}\in \mathbb{C}^{M\times \tilde{S}}$, where $\tilde{S}=\sum_{k=1}^{K}S_k$. Any training matrix $\mat{S}_{\mathrm{DL}}$ can be decomposed as
\begin{align*}
\mat{S}_{\mathrm{DL}}&= \mat{S}_{\mathrm{DL}}^{\parallel}+\mat{S}_{\mathrm{DL}}^{\perp},
\end{align*}
where the columns of $\mat{S}_{\mathrm{DL}}^{\parallel}$ and $\mat{S}_{\mathrm{DL}}^{\perp}$ belong to the column space of $\tilde{\mat{Q}}$ and its null space respectively. They can be respectively seen as the training power transmitted in the direction and in the nulls of all users.

\begin{proposition} \label{proposition:general_form}
	A necessary condition for an optimal training design is that $\|\mat{S}_{\mathrm{DL}}^{\perp}\|={0}$. Equivalently, it implies that an optimal training matrix $\mat{S}_{\mathrm{DL}}$ has to have the form $\mat{S}_{\mathrm{DL}}=\tilde{\mat{Q}}\mat{X}$, for some matrix $\mat{X}\in \mathbb{C}^{\tilde{S}\times T}$. 
\end{proposition}
\begin{proof}
	Let us consider a training matrix $\mat{S}_{\mathrm{DL}}$ such that $\|\mat{S}_{\mathrm{DL}}^{\perp}\|\geq {0}$. Setting $\|\mat{S}_{\mathrm{DL}}^{\perp}\|$ to zero, does not affect the capacity of any user while it reduces the training power. This reduction of  power could be used, \textit{e.g.}, to amplify $\mat{S}_{\mathrm{DL}}^{\parallel}$ and improve the capacity of other users.
\end{proof}
Intuitively, this optimality condition implies that no training power is wasted in null directions.

\subsection{Concavity}

\begin{proposition} \label{proposition:concavity}
	 The user secret-key capacity $C_k$ given in $(\ref{eq:C_k})$ is concave in the training covariance matrix $\mat{C}_{\mathrm{DL}}\succcurlyeq \mat{0}$.
\end{proposition}

\begin{proof}
	Due to space constraints, see extended version to appear soon.
%
\end{proof}

The concavity of $C_k$ implies that multiple problems can be solved efficiently using numerical optimization. For instance, one can efficiently maximize the sum secret-key capacity
\begin{align}
	\max_{\mat{C}_{\mathrm{DL}}\succcurlyeq \mat{0}} \ \sum_{k=1}^K C_k \quad \text{s.t. } \tr[\mat{C}_{\mathrm{DL}}]\leq p_{\mathrm{DL}}, \label{eq:max_sum_capa}
\end{align}
or maximize the minimal secret-key capacity
\begin{align}
\max_{\mat{C}_{\mathrm{DL}}\succcurlyeq \mat{0}} \ \min_k \ C_k \quad \text{s.t. }\tr[\mat{C}_{\mathrm{DL}}]\leq p_{\mathrm{DL}}. \label{eq:max_min_capa}
\end{align}
Note that one can easily give priorities to certain users by considering weighted versions of (\ref{eq:max_sum_capa}) and (\ref{eq:max_min_capa}), \textit{i.e.}, assigning priorities $\beta_k$ to each user so that $\sum_{k=1}^K\beta_k=1$ and $C_k$ is replaced by $C_k\beta_k$ in (\ref{eq:max_sum_capa}) and $C_k/\beta_k$ in (\ref{eq:max_min_capa}).


\subsection{Large Antenna Case}

As the number of BS antennas grows large, the spatial resolution of the BS increases. Hence, the BS becomes able to focus on and discriminate multipath components coming from different directions. As a result, considering an arbitrary pair of spatial modes corresponding to different users, they become orthogonal if coming from a different spatial direction. This assumption is now mathematically formalized.

$\mathbf{(As1)}$: the spatial modes of different users come from different spatial directions. This implies that, as $M/K \rightarrow \infty$, $\forall k,k',s,s'$, $|\vect{q}_{s,k}^\dagger \vect{q}_{s',k'}|\rightarrow 0$.

\begin{theorem} \label{theorem:large_antenna_case_1}
	Under $\mathbf{(As1)}$, the following training matrix, with a number of pilots limited to only $T=\max_k(S_k)$, has an optimal structure
	\begin{align}
	\mat{S}_{\mathrm{DL}}&=\sum_{k=1}^K \mat{Q}_k \begin{pmatrix}
	\mat{P}_{k,\mathrm{DL}}^{1/2}& \mat{0}_{T-S_k\times T}
	\end{pmatrix} \in \mathbb{C}^{M\times T}, \label{eq:opt_training_sequence_As1}
	\end{align}
	where $\mat{P}_{k,\mathrm{DL}}^{1/2}\in \mathbb{C}^{S_k \times S_k}$ is a diagonal matrix. The $s$-th diagonal element, denoted by $p_{s,k,\mathrm{DL}}$, corresponds to the power associated to the $s$-th spatial mode of user $k$. The resulting secret-key capacity for user $k$ is
	\begin{align}
	C_k 
	&= \sum_{s=1}^{S_k} \log \left(1+ \frac{\lambda_{s,k}^2\rho_{k,\mathrm{UL}}\rho_{s,k,\mathrm{DL}}}{\lambda_{s,k}(\rho_{k,\mathrm{UL}}+\rho_{s,k,\mathrm{DL}})+1 }\right), \label{eq:C_k_As1}
	\end{align}
	where $\rho_{s,k,\mathrm{DL}}=\frac{{p}_{s,k,\mathrm{DL}}}{\sigma_{\mathrm{DL}}^2}$. The sum secret-key capacity optimization (\ref{eq:max_sum_capa})
	is solved by using the training sequence (\ref{eq:opt_training_sequence_As1}) together with the water filling power allocation $p_{s,k,\mathrm{DL}}=\left(\mu -\sigma_{\mathrm{DL}}^2/{\lambda_{s,k}}\right)^+$, where $\mu$ is a positive constant that ensures that the power constraint is satisfied $\tr[\mat{C}_{\mathrm{DL}}]= \sum_{k=1}^K\tr[\mat{P}_{k,\mathrm{DL}}]= p_{\mathrm{DL}}$.
\end{theorem}
\begin{proof}
	Due to space constraints, see extended version to appear soon.
\end{proof}

The theorem provides significant insights for large antenna systems. The assumption $\mathbf{(As1)}$ decouples and simplifies the training design as the problem becomes separable between users. Sending only $T=\max_k(S_k)$ pilots becomes sufficient to reach the capacity, implying that it does not scale with the number of users being served. Secondly, the theorem gives the structure of the optimal training sequence. Only the power associated to each spatial mode ${p}_{s,k,\mathrm{DL}}$ has to be found, depending on the criterion to maximize. In the case of the sum capacity, a water filling solution is found again.

Moreover, as in the single-user case, it might be useful to look for training designs with a reduced number of pilots. The following corollary shows that, under $\mathbf{(As1)}$, a constraint on the maximal number of pilots can be easily taken into account.
\begin{corollary} \label{corollary:MU_pilot_constraint}
	Under $\mathbf{(As1)}$, if a constraint is added on the number of pilots $T\leq T_{\mathrm{max}}$, the optimal training design is the same as in (\ref{eq:opt_training_sequence_As1}) and only the $T_{\mathrm{max}}$ strongest spatial modes of each user will be active. The sum secret-key capacity optimization (\ref{eq:max_sum_capa}) is then solved by a water filling solution over the active spatial modes.
\end{corollary}
\begin{proof}
	Due to space constraints, see extended version to appear soon.
\end{proof}

As a final note, one can check that the single-user results of Theorem~\ref{theorem:proof_theorem_SU} and Corollary~\ref{corollary:SU_pilot_constraint} are found back as a particularization of Theorem~\ref{theorem:large_antenna_case_1} and Corollary~\ref{corollary:MU_pilot_constraint} to the case $K=1$. In that case, $\mathbf{(As1)}$ is trivially always verified.


\section{Simulation Results}

\begin{figure}[!t]
	\centering
	\resizebox{0.45\textwidth}{!}{%
		\Large
%
%
\begin{tikzpicture}

\begin{axis}[%
width=4.520833in,
height=3.565625in,
at={(0.758333in,0.48125in)},
scale only axis,
xmin=0,
xmax=20,
xlabel={Downlink SNR $\rho_{\mathrm{DL}}$ [dB]},
ymin=0,
ymax=30,
ylabel={Average secret-key capacity $\bar{C}$ [bits/training]},
title={$M=32$ BS antennas, Uplink SNR $\rho_{\mathrm{UL}}=10$ dB},
legend style={at={(0.03,0.97)},anchor=north west,legend cell align=left,align=left,draw=white!15!black}
]
\addplot [color=black,dashed,line width=1.5pt]
  table[row sep=crcr]{%
0	1.2764031716487\\
5	2.78480913845306\\
10	5.33033718502128\\
15	9.29018500440832\\
20	14.6850543946262\\
};
\addlegendentry{Uniform alloc.};

\addplot [color=black,dotted,line width=1.5pt]
  table[row sep=crcr]{%
0	3.88228472228069\\
5	5.94745160917685\\
10	8.61844294475646\\
15	12.4221509963221\\
20	17.1641870323472\\
};
\addlegendentry{Large antenna alloc. [Th. 2]};

\addplot [color=black,solid,line width=1.5pt]
  table[row sep=crcr]{%
0	4.31632531692653\\
5	6.40401456740039\\
10	9.08602714016107\\
15	12.8900070285179\\
20	17.6236770561086\\
};
\addlegendentry{Optimal alloc.};

\addplot [color=black,dashed,line width=1.5pt,forget plot]
  table[row sep=crcr]{%
0	1.57893168601656\\
5	3.67569251291566\\
10	7.67761603258746\\
15	14.5445697322249\\
20	24.024191296855\\
};
\addplot [color=black,solid,line width=1.5pt,forget plot]
  table[row sep=crcr]{%
0	4.73967655163156\\
5	7.59937347861168\\
10	11.9961311753365\\
15	18.8175045694379\\
20	27.2799162173267\\
};
\end{axis}
\node at (8.5,6.9) {Single-user};
\node at (8.5,6.3) {$K=1$};
\draw (10.5,6.3) circle [radius=0.8];

\node at (11.8,3.6) {Multiple-user};
\node at (12,3) {$K=2$};
\draw (12,5.4) circle [radius=0.7];

\end{tikzpicture}%
	}
		\vspace{-0.5em}
	\caption{Average secret-key capacity as a function of downlink SNR for different training sequences.}
	\label{fig:Capacity_SNR}
		\vspace{-1em}
\end{figure}
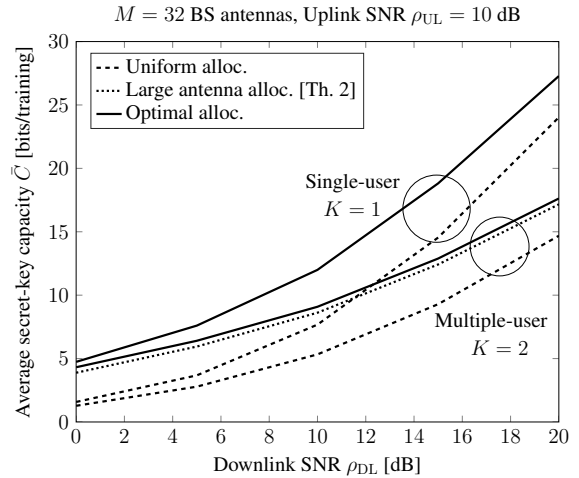

\begin{figure}[!t]
	\centering
	\resizebox{0.45\textwidth}{!}{%
		\Large
%
%
\definecolor{mycolor1}{rgb}{0.20810,0.16630,0.52920}%
\definecolor{mycolor2}{rgb}{0.21783,0.72504,0.61926}%
\definecolor{mycolor3}{rgb}{0.97630,0.98310,0.05380}%
\begin{tikzpicture}

\begin{axis}[%
width=4.520833in,
height=3.565625in,
at={(0.758333in,0.48125in)},
scale only axis,
area legend,
xmin=-5,
xmax=25,
xtick={ 0,  5, 10, 15, 20},
xlabel={Downlink SNR $\rho_{\mathrm{DL}}$ [dB]},
ymin=0,
ymax=14,
ylabel={Number of pilots $T$},
title={$M=32$ BS antennas, Uplink SNR $\rho_{\mathrm{UL}}=10$ dB},
legend style={at={(0.03,0.97)},anchor=north west,legend cell align=left,align=left,draw=white!15!black}
]
\addplot[ybar,bar width=0.133951in,bar shift=-0.167438in,draw=black,fill=mycolor1] plot table[row sep=crcr] {%
0	3\\
5	4\\
10	9\\
15	11\\
20	14\\
};
\addlegendentry{Optimal alloc. $K=1$};

\addplot[ybar,bar width=0.133951in,draw=black,fill=mycolor2] plot table[row sep=crcr] {%
0	2\\
5	3\\
10	7\\
15	11\\
20	14\\
};
\addlegendentry{Optimal alloc. $K=2$};

\addplot[ybar,bar width=0.133951in,bar shift=0.167438in,draw=black,fill=mycolor3] plot table[row sep=crcr] {%
0	2\\
5	4\\
10	7\\
15	11\\
20	14\\
};
\addlegendentry{Large antenna [Th. 2] $K=2$};

\end{axis}
\end{tikzpicture}%
	}
		\vspace{-0.5em}
	\caption{Number of pilots as a function of downlink SNR for different training sequences.}
	\label{fig:Pilot_number_SNR}
		\vspace{-1.5em}
\end{figure}
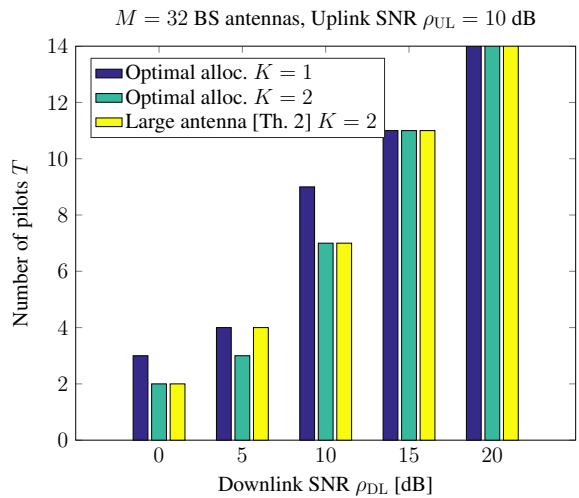

This sections aims at analyzing the theoretical results derived in previous sections for realistic channel propagation conditions. The carrier frequency is set to 3.5 GHz. A rectangular planar array of antennas at receive side is considered with an inter-antenna spacing of half a wavelength. The antenna elements have an isotropic pattern with ideal vertical polarization. The channel parameters were generated by QuaDRiGa \cite{Jaeckel2014} according to the 3D-UMa NLOS model defined by 3GPP TR 36.873 v12.5.0 specifications \cite{3GPP_TR_36_873v12_5_0}. Moreover, the channel parameters are generated for $K=2$ users locations randomly drawn in a radius of 200 meters around the BS; the BS height is 20m above ground. The same set of channel parameters was used for all simulations. 

In Fig.~\ref{fig:Capacity_SNR}, the average secret-key capacity, defined as
\begin{align*}
	\bar{C}=\frac{1}{K}\sum_{k=1}^KC_k,
\end{align*}
is plotted as a function of the downlink SNR $\rho_{\mathrm{DL}}$ for different training sequences. An uplink SNR $\rho_{\mathrm{UL}}=10$ dB is considered and a BS equipped with $M=32$ antennas placed in a rectangular fashion: $8\text{ Horiz.}\times 4\text{ Vert.}$. Both the single-user ($K=1$) and the multiple-user cases ($K=2$) are considered. The large antenna allocation corresponds to the training sequence given in (\ref{eq:opt_training_sequence_As1}), where the power allocation is found through the waterfilling algorithm. The uniform allocation corresponds also to the training sequence given in (\ref{eq:opt_training_sequence_As1}) but with uniform power allocated over the spatial modes. The optimal allocation corresponds to the design of Theorem~\ref{theorem:proof_theorem_SU} in the single-user case while, in the multiple-user case, the problem (\ref{eq:max_sum_capa}) is solved using CVX \cite{cvx} together with the MOSEK solver.

In the single-user case, as shown in Theorem~\ref{theorem:proof_theorem_SU}, the large antenna allocation becomes equivalent to (\ref{eq:SU_opt_training}) and is thus optimal, which is why they are mixed in Fig.~\ref{fig:Capacity_SNR}. On the other hand, the large antenna allocation is not optimal in the multiple-user case but the performance gap with the optimal design is rather small, given the large number of antennas. This motivates the use of this design in practice for large values of $M/K$. Finally, as expected, one can see that the performance loss of the uniform allocation reduces with the SNR. Indeed, as the SNR increases, the waterfilling solution results in a larger number of active modes and asymptotically equal power allocated to each mode. As a final note, the average secret-key capacity is smaller in the multiple-user case than in the single-user case. This can be intuitively expected as, for the same power budget, the BS has to accommodate a larger number of users. Of course, the overall sum secret-key capacity $K\bar{C}$ (not plotted in the figure) remains still larger than in the single-user case.

In Fig.~\ref{fig:Pilot_number_SNR}, for the same scenario as considered in Fig.~\ref{fig:Capacity_SNR}, the length of the training sequence, \textit{i.e.}, $T$, is plotted as a function of the downlink SNR. As expected from Theorem~\ref{theorem:large_antenna_case_1}, going from the single-user case ($K=1$) to the multiple-user case ($K=2$) does not result in an increased number of pilots. Interestingly, the number of pilots is even smaller at certain SNR values. All of this is a positive news as, in the end, the required number of pilots for an optimal training sequence is far from the number of BS antennas $M$ and does not scale with the number of users $K$.

%
%

\vspace{-1em}
\section{Conclusion} \label{section:conclusion}

This paper studied the optimization of the downlink training sequence for maximizing the secret-key capacity in a TDD massive MIMO scenario. Massive gains can be obtained by leveraging spatial dimensionality gains and array gains. Closed-form training designs and expressions of the secret-key capacity are provided, both in the single-user and multiple-user cases. A positive result is that, in the large antenna case, the pilot overhead does not scale with the number of served users. The optimization can take into account constraints on the number of pilots and various criteria such as the maximization of the minimal capacity and the sum capacity. User priorities can also be easily accommodated.

\vspace{-0.5em}
\scriptsize 
\bibliographystyle{IEEEtran}
\bibliography{IEEEabrv,IEEEreferences}

\end{document}